\documentclass[reqno]{amsart}
\usepackage{amssymb}

\newtheorem{theorem}{Theorem}

\theoremstyle{remark}
\newtheorem{remark}[theorem]{Remark}

\numberwithin{equation}{section}

\begin{document}

\title[hyperbolic Ruijsenaars-Schneider system with Morse term]
{Spectrum and eigenfunctions of the lattice hyperbolic Ruijsenaars-Schneider system with exponential Morse term}

\author{J.F.  van Diejen}

\address{
Instituto de Matem\'atica y F\'{\i}sica, Universidad de Talca,
Casilla 747, Talca, Chile}

\email{diejen@inst-mat.utalca.cl}

\author{E. Emsiz}

\address{
Facultad de Matem\'aticas, Pontificia Universidad Cat\'olica de Chile,
Casilla 306, Correo 22, Santiago, Chile}
\email{eemsiz@mat.puc.cl}


\thanks{This work was supported in part by the {\em Fondo Nacional de Desarrollo
Cient\'{\i}fico y Tecnol\'ogico (FONDECYT)} Grants \# 1130226 and  \# 1141114.}

\date{August 2015}

\begin{abstract}
We place the hyperbolic quantum Ruijsenaars-Schneider system with an exponential Morse term on a lattice and diagonalize the resulting $n$-particle model
by means of multivariate continuous dual $q$-Hahn polynomials that arise as a parameter
reduction of the Macdonald-Koornwinder polynomials. This allows to  compute the
$n$-particle scattering operator, to identify the bispectral dual system, and to confirm the quantum integrability in a Hilbert space set-up.
\end{abstract}

\maketitle

\section{Introduction}\label{sec1}
It is well-known that the hyperbolic Calogero-Moser $n$-particle system on the line can be placed in an exponential Morse potential without spoiling the integrability \cite{adl:some,lev-woj:olshanetsky-perelomov}.   An extension of Manin's Painlev\'e-Calogero correspondence
links the particle model in question to a multicomponent Painlev\'e III equation \cite{tak:painleve-calogero}.
Just as for the conventional Calogero-Moser system without Morse potential, the integrability is preserved upon quantization and the corresponding spectral problem gives rise to a rich theory of remarkable novel hypergeometric functions  in several variables \cite{ino-mes:discrete,osh:completely,hal:multivariable,hal-lan:unified}. 

An integrable Ruijsenaars-Schneider type ($q$-)deformation \cite{rui-sch:new,rui:complete} of the hyperbolic Calogero-Moser system with Morse potential was introduced in
\cite{sch:integrable} and in a more general form in \cite[Sec. II.B]{die:difference}. 
Recently, it was pointed out that particle systems of this kind can be recovered from the Heisenberg double of $SU(n,n)$ via Hamiltonian reduction
 \cite{mar:new}. 
 In the present work we  address the eigenvalue problem for a quantization of the latter
 hyperbolic Ruijsenaars-Schneider system with Morse term.
 Specifically, it is shown that the eigenfunctions are given by multivariate continuous dual $q$-Hahn polynomials that arise as a parameter
reduction of the Macdonald-Koornwinder polynomials \cite{koo:askey-wilson,mac:affine}. As immediate
by-products, one reads-off the $n$-particle scattering operator and the
commuting quantum integrals of a bispectral dual system  \cite{dui-gru:differential,gru:bispectral}.

The material is organized is as follows. In Section \ref{sec2} we place the
hyperbolic Ruijsenaars-Schneider system with Morse term from \cite{die:difference} on a lattice. 
The diagonalization of the resulting quantum model in terms of multivariate continuous dual $q$-Hahn polynomials is carried out in Section \ref{sec3}.
In Sections \ref{sec4} and \ref{sec5} the $n$-particle scattering operator and the bispectral dual integrable system are exhibited.
Finally, the quantum integrability of both the hyperbolic Ruijsenaars-Schneider system with Morse term on the lattice and its bispectral dual system are addressed in Section \ref{sec6}.

\section{Hyperbolic Ruijsenaars-Schneider system with Morse term}\label{sec2}
The hyperbolic quantum Ruijsenaars-Schneider system on the lattice was briefly introduced in \cite[Sec. 3C2]{rui:finite-dimensional} and studied in detail from the point of view of its scattering behavior in \cite{rui:factorized}  (see also \cite[Sec. 6]{die:scattering} for a further generalization in terms of root systems).
In this section we formulate a corresponding lattice version of the hyperbolic quantum Ruijsenaars-Schneider system with Morse term introduced in \cite[Sec. II.B]{die:difference}.

\subsection{Hamiltonian}
The Hamiltonian of our $n$-particle model is given by the formal difference operator
\cite[Eqs. (2.25), (2.26)]{die:difference}:
\begin{align}\label{H}
H:=  \sum_{j=1}^n  \Biggl( & \ w_+(x_j) \Bigl(\prod_{\substack{1\leq k\leq n \\ k\neq j}} \frac{t^{-1}-q^{x_j-x_k}}{1-q^{x_j-x_k}} \Bigr)   (T_j-1) \\
 +& \ w_-(x_j) \Bigl(\prod_{\substack{1\leq k\leq n \\ k\neq j}} \frac{t-q^{x_j-x_k}}{1-q^{x_j-x_k}} \Bigr)   (T_j^{-1}-1) \Biggr) ,\nonumber
\end{align}
where
\begin{align}
w_+ (x) &:=\sqrt{\frac{qt_0t_3}{t_1t_2} } (1-{t}_1 q^{x})(1-t_2q^x),\quad
w_- (x) := \sqrt{\frac{t_1t_2}{qt_0t_3}} (1- t_0q^{x})(1-t_3q^x) ,\nonumber
\end{align}
and $T_j$ ($j=1,\ldots ,n$) acts on functions $f:\mathbb{R}^n\to\mathbb{C}$ by a
unit translation of the $j$th position variable $$(T_jf)(x_1,\ldots,x_n)=f(x_1,\ldots,x_{j-1},x_j+1,x_{j+1},\ldots,x_n).$$
Here $q$  denotes a real-valued scale parameter, $t$ plays the role of the coupling parameter for the Ruijsenaars-Schneider inter-particle interaction, and
the parameters ${t}_r$ ($r=0,\ldots ,3$) are coupling parameters governing the exponential Morse interaction.
Upon setting $t_0=\epsilon t^{n-1}q^{-1}$ and $t_r=\epsilon$ for $r=1,2 ,3$,  one has that
$w_\pm(x_j)\to t^{\pm (n-1)/2}$ when $\epsilon\to 0$. We thus recover in this limit the Hamiltonian of the hyperbolic quantum Ruijsenaars-Schneider system given in terms of
Ruijsenaars-Macdonald difference operators \cite{rui:complete,mac:symmetric}. 
By a translation of the center-of-mass of the form $q^{x_j}\to cq^{x_j}$ ($j=1,\ldots ,n$) for some suitable constant $c$, it is possible to normalize one of the $t_r$-parameters to unit value; from now on it will therefore always be assumed that $t_3\equiv 1$ unless explicitly stated otherwise.

\subsection{Restriction to lattice functions}
Let $\rho+\Lambda:=\{\rho +\lambda\mid\lambda\in\Lambda\}$, where 
$\Lambda$ denotes the cone of integer partitions $\lambda=(\lambda_1,\ldots ,\lambda_n)$  with
weakly decreasingly ordered parts $\lambda_1\geq\cdots\geq\lambda_n\geq 0$, and
$\rho=(\rho_1,\ldots ,\rho_n)$ with
 \begin{equation}\label{rho}
 \rho_j=(n-j)\log_q( t)\qquad  (j=1,\ldots,n). 
 \end{equation}
 The action of $H$ \eqref{H}
(with $t_3=1$) preserves the space of lattice functions $f:\rho+\Lambda \to\mathbb{C}$:
  \begin{align}\label{Haction}
(Hf)(\rho +\lambda)=& 
\sum_{\substack{1\leq j\leq n\\ \lambda +e_j\in\Lambda}} v_j^+(\lambda) \bigl( f (\rho+\lambda +e_j)-f(\rho+\lambda) \bigr)\\
+& \sum_{\substack{1\leq j\leq n\\ \lambda -e_j\in\Lambda}}v_j^-(\lambda)\bigl( f (\rho+\lambda -e_j) -f(\rho+\lambda)\bigr) 
 ,\nonumber
\end{align}
where $e_1,\ldots ,e_n$ denotes the standard basis of $\mathbb{R}^n$ and
\begin{align*}
v_j^+(\lambda) =&\sqrt{\frac{qt_0}{t_1t_2} } (1-{t}_1 t^{n-j}q^{\lambda_j})(1-t_2t^{n-j}q^{\lambda_j})
\prod_{\substack{1\leq k\leq n \\ k\neq j}} \frac{t^{-1}-t^{k-j}q^{\lambda_j-\lambda_k}}{1-t^{k-j}q^{\lambda_j-\lambda_k}}  , \\
v_j^-(\lambda) =&\sqrt{\frac{t_1t_2}{qt_0}} (1- t_0t^{n-j}q^{\lambda_j})(1-t^{n-j}q^{\lambda_j}) \prod_{\substack{1\leq k\leq n \\ k\neq j}} \frac{t-t^{k-j}q^{\lambda_j-\lambda_k}}{1-t^{k-j}q^{\lambda_j-\lambda_k}} .
\end{align*}
Indeed, given $\lambda\in\Lambda$, one has that $v_j^+(\lambda)=0$ if $\lambda+e_j\not\in \Lambda$ due to a zero stemming from the
factor $t^{-1}-t^{-1}q^{\lambda_{j-1}-\lambda_j}$ when $\lambda_{j-1}=\lambda_j$ and one has that
 $v_j^-(\lambda)=0$ if $\lambda-e_j\not\in \Lambda$ due to a zero stemming from either the
factor $t-t q^{\lambda_{j}-\lambda_{j+1}}$ when $\lambda_j=\lambda_{j+1}$ or from the factor $(1-q^{\lambda_n})$ when  $\lambda_n=0$.

\section{Spectrum and eigenfunctions}\label{sec3}
Ruijsenaars' starting point in \cite{rui:factorized} is the fact that the hyperbolic quantum Ruijsenaars-Schneider system on the lattice is
diagonalized by the celebrated Macdonald polynomials \cite[Ch.VI]{mac:symmetric}. In this section we show that in the presence of the Morse interaction the role of the Macdonald eigenpolynomials is taken over by multivariate continuous dual $q$-Hahn eigenpolynomials that arise as a parameter reduction of the Macdonald-Koornwinder polynomials \cite{koo:askey-wilson,mac:affine}.

\subsection{Multivariate continuous dual $q$-Hahn polynomials}
Continuous dual $q$-Hahn polynomials are a special limiting case of the Askey-Wilson polynomials in which one of the four Askey-Wilson parameters is set to vanish \cite[Ch. 14.3]{koe-les-swa:hypergeometric}. The corresponding reduction of the Macdonald-Koornwinder multivariate Askey-Wilson polynomials \cite{koo:askey-wilson,mac:affine} is governed  by a weight function of the form
\begin{equation}\label{plancherel}
\hat{\Delta} (\xi ):=\frac{1}{(2\pi)^n}
\prod_{1\leq j\leq n}\Bigl|
\frac{(e^{2i\xi_j})_\infty}{\prod_{0\leq r\leq 2} (\hat{t}_re^{i\xi_j})_\infty}  \Bigr|^2 
\prod_{1\leq j<k\leq n}\Bigl| 
\frac{(e^{i(\xi_j+\xi_k)},e^{i(\xi_j-\xi_k)})_\infty}{(te^{i(\xi_j+\xi_k)},te^{i(\xi_j-\xi_k)})_\infty}\Bigr|^2
\end{equation}
supported on the alcove
\begin{equation}\label{alcove}
\mathbb{A}:=\{ (\xi_1,\xi_2,\ldots,\xi_n)\in\mathbb{R}^n\mid \pi>\xi_1>\xi_2>\cdots >\xi_n>0\} ,
\end{equation}
where $(x)_m:=\prod_{l =0}^{m-1} (1-xq^l)$ and $(x_1,\ldots, x_l)_m:=(x_1)_m\cdots (x_l)_m$ refer to the $q$-Pochhammer symbols, and  it  is assumed that
\begin{equation}\label{pdomain}
q,t\in (0,1)\quad\text{and}\quad \hat{t}_r\in (-1,1)\setminus \{ 0\} \quad (r=0,1 ,2).
\end{equation}
More specifically, the multivariate continuous dual $q$-Hahn polynomials $P_\lambda (\xi)$, $\lambda\in\Lambda$ are defined as the trigonometric polynomials of the form
\begin{subequations}
\begin{equation}\label{qH1}
P_\lambda(\xi) = \sum_{\substack{\mu\in\Lambda\\ \mu \leq\lambda}}
c_{\lambda ,\mu}m_\mu(\xi)\qquad (c_{\lambda ,\mu}\in\mathbb{C})
\end{equation}
such that
\begin{equation}\label{qH2}
c_{\lambda ,\lambda} =
\prod_{1\leq j\leq n} \frac{\hat{t}_0^{\lambda_j}t^{(n-j)\lambda_j}}{ (\hat{t}_0\hat{t}_1t^{n-j},\hat{t}_0\hat{t}_2 t^{n-j} )_{\lambda_j}}
\prod_{1\leq j<k\leq n}
\frac{(t^{k-j})_{\lambda_j-\lambda_k}}{(t^{1+k-j})_{\lambda_j-\lambda_k}}
\end{equation}
and
\begin{equation}\label{qH3}
\int_{\mathbb{A}} P_\lambda (\xi) \overline{P_\mu (\xi)}
\hat{\Delta}(\xi) \text{d}\xi= 0\quad \text{if}\ \mu <\lambda .
\end{equation}
\end{subequations}
Here we have employed the dominance partial order
\begin{equation}\label{porder}
\forall \mu,\lambda\in\Lambda:\quad \mu\leq \lambda \ \text{iff}\
\sum_{1\leq j\leq k} \mu_j \leq \sum_{1\leq j\leq k} \lambda_j \quad\text{for}\quad k=1,\ldots ,n,
\end{equation}
and the symmetric monomials
\begin{equation}
m_\lambda(\xi):=\sum_{\nu\in W\lambda} e^{i(\nu_1\xi_1+\cdots +\nu_n\xi_n)},\qquad \lambda\in\Lambda,
\end{equation}
 associated with the hyperoctahedral group
$W=S_n \ltimes \{1,-1\}^n$ of signed permutations.

The present choice of the leading coefficient $c_{\lambda ,\lambda}$ in Eq. \eqref{qH2} normalizes the polynomials in question such that
$P_\lambda (i\hat{\rho})=1$, where $\hat{\rho}=(\hat{\rho}_1,\ldots ,\hat{\rho}_n)$ is given by $\hat{\rho}_j = (n-j)\log (t) +\log (\hat{t}_0)$, $j=1,\ldots ,n$ (cf. \cite[Sec. 6]{die:properties}, \cite[Ch. 5.3]{mac:affine}).
With this normalization, the orthogonality relations obtained as the degeneration of those for the
Macdonald-Koornwinder polynomials \cite[\text{Sec.}~5]{koo:askey-wilson}, \cite[\text{Sec.}~7]{die:properties}, \cite[\text{Ch.}~5.3]{mac:affine} read:
\begin{subequations}
\begin{equation}\label{qH-ortho-a}
\int_{\mathbb{A}} P_\lambda (\xi) \overline{ P_\mu (\xi)}
\hat{\Delta}(\xi) \text{d}\xi =
\begin{cases}
\Delta_\lambda^{-1} &\text{if}\ \lambda = \mu ,\\
0 &\text{otherwise},
\end{cases}
\end{equation}
where
\begin{align}\label{qH-ortho-b}
\Delta_\lambda  :=&\Delta_0 
 \prod_{1\leq j\leq n}
\frac{(\hat{t}_0\hat{t}_1t^{n-j},\hat{t}_0\hat{t}_2t^{n-j})_{\lambda_j}} {\hat{t}_0^{2\lambda_j}t^{2(n-j)\lambda_j}(qt^{n-j},\hat{t}_1\hat{t}_2t^{n-j})_{\lambda_j}}    \\
&\times \prod_{1\leq j<k\leq n}
\frac{1-t^{k-j}q^{\lambda_j-\lambda_k}}{1-t^{k-j}}
\frac{(t^{1+k-j})_{\lambda_j-\lambda_k}}{(qt^{k-j-1})_{\lambda_j-\lambda_k}} 
\nonumber
\end{align}
and
\begin{equation}\label{qH-ortho-c}
\Delta_0:=
\prod_{1\leq j\leq n} \Bigl( \frac{(q,t^j)_\infty }{(t)_\infty} \prod_{0\leq r<s\leq 2} (\hat{t}_r\hat{t}_st^{n-j})_\infty \Bigr) .
\end{equation}
\end{subequations}

\subsection{Diagonalization}
Let
$\ell^2(\rho+\Lambda,{\Delta})$ denote the Hilbert space of lattice functions ${f}:\rho+\Lambda\to\mathbb{C}$ determined by the inner product
\begin{equation}
\langle {f},{g}\rangle_{{\Delta}}:=
\sum_{\lambda\in\Lambda}{f}(\rho+\lambda)\overline{{g}(\rho+\lambda)}{\Delta}_\lambda
\qquad ({f},{g}\in \ell^2(\rho+\lambda,{\Delta})),
\end{equation}
with $\rho$ and $\Delta_\lambda$ as in Eqs.  \eqref{rho}  and \eqref{qH-ortho-a}--\eqref{qH-ortho-c}, and let $L^2(\mathbb{A},\hat{\Delta} (\xi) \text{d}\xi)$ be the Hilbert space of functions $\hat{f}:\mathbb{A}\to\mathbb{C}$ determined by the inner product
\begin{equation}\label{ip}
\langle \hat{f},\hat{g}\rangle_{\hat{\Delta}}:=\int_{\mathbb{A}}\hat{f}(\xi)\overline{\hat{g}(\xi)}\hat{\Delta}(\xi)\text{d}\xi
\qquad (\hat{f},\hat{g}\in L^2(\mathbb{A},\hat{\Delta} (\xi) \text{d}\xi)),
\end{equation}
with $\hat{\Delta}$ taken from Eq. \eqref{plancherel}. We denote by $\psi_\xi:\rho+\Lambda\to\mathbb{C}$  the lattice wave function given by
\begin{equation}\label{fk}
\psi_\xi (\rho+\lambda):=P_\lambda (\xi) \qquad (\xi\in\mathbb{A},\,\lambda\in\Lambda).
\end{equation}
Then the orthogonality relations in Eqs. \eqref{qH-ortho-a}--\eqref{qH-ortho-c} imply that the associated Fourier transform
 $\boldsymbol{F}: \ell^2(\rho+\Lambda,{\Delta})\to L^2(\mathbb{A},\hat{\Delta}\text{d}\xi)$
of the form
 \begin{subequations}
\begin{equation}\label{ft1}
(\boldsymbol{F}{f})(\xi):= \langle {f},\psi_\xi \rangle_{{\Delta}}=\sum_{\lambda\in\Lambda}f(\rho+\lambda)
\overline{\psi_\xi (\rho +\lambda )}{\Delta}_\lambda
\end{equation}
(${f}\in \ell^2(\rho+\Lambda,{\Delta})$) constitutes a Hilbert space isomorphism
with an inversion formula given by
\begin{equation}\label{ft2}
(\boldsymbol{F}^{-1}\hat{f})(\rho+\lambda) = \langle \hat{f},\overline{\psi (\rho+\lambda) }\rangle_{\hat{\Delta}}=
\int_{\mathbb{A}} \hat{f} (\xi) \psi_\xi(\rho +\lambda )\hat{\Delta}(\xi)\text{d}\xi
\end{equation}
\end{subequations}
($\hat{f}\in L^2(\mathbb{A},\hat{\Delta}\text{d}\xi)$). 

\begin{theorem}\label{diagonal:thm}
Let $\hat{{E}}$ denote the bounded real multiplication operator acting on
$\hat{f}\in L^2(\mathbb{A},\hat{\Delta}\text{d}\xi)$ by $(\hat{{E}}\hat{f})(\xi):=\hat{E}(\xi) \hat{f}(\xi)$ with
\begin{subequations}
\begin{equation}\label{E}
\hat{E}(\xi):=
\sum_{1\leq j\leq n}\bigl( 2\cos(\xi_j) -t^{n-j}\hat{t}_0-t^{j-n}\hat{t}_0^{-1}\bigr) .
\end{equation}
For 
\begin{equation}\label{pms}
t_0=q^{-1}\hat{t}_1\hat{t}_2,\quad t_1=\hat{t}_0\hat{t}_2, \quad t_2=\hat{t}_0\hat{t}_1 
\end{equation}
with $q,t$ and $\hat{t}_r$ in the parameter domain \eqref{pdomain}, the hyperbolic lattice Ruijsenaars-Schneider
Hamiltonian with Morse interaction $H$ \eqref{Haction}  constitutes a bounded self-adjoint operator in the Hilbert space $\ell^2(\rho+\Lambda,\Delta)$  
diagonalized by the Fourier transform $\boldsymbol{F}$ \eqref{ft1}, \eqref{ft2}:
 \begin{equation}\label{Hdiagonal}
H=\boldsymbol{F}^{-1}  \circ \hat{{E}} \circ\boldsymbol{F} .
 \end{equation}
\end{subequations}
\end{theorem}

\begin{proof}
It suffices to verify that the Fourier kernel $\psi_\xi$ \eqref{fk} satisfies
the eigenvalue equation $H\psi_\xi=\hat{E}(\xi)\psi_\xi$, or more explicitly that:
\begin{align*}
 &\sum_{\substack{1\leq j\leq n\\ \lambda +e_j\in\Lambda}} v_j^+(\lambda) \bigl( \psi_\xi (\rho+\lambda +e_j)-\psi_\xi (\rho+\lambda )\bigr) + \\
 &\sum_{\substack{1\leq j\leq n\\ \lambda -e_j\in\Lambda}}v_j^-(\lambda)\bigr( \psi_\xi (\rho+\lambda -e_j)-\psi_\xi (\rho+\lambda )\bigr) =\hat{E}(\xi)\psi_\xi (\rho +\lambda) . \nonumber
\end{align*}
This eigenvalue equation amounts to the continuous dual $q$-Hahn reduction of
the Pieri recurrence formula for the Macdonald-Koornwinder polynomials corresponding to Eqs. (6.4), (6.5) and Section 6.1 of \cite{die:properties}.
\end{proof}
It is immediate from Theorem \ref{diagonal:thm} that the hyperbolic lattice Ruijsenaars-Schneider
Hamiltonian with Morse interaction $H$ \eqref{Haction} 
has purely absolutely continuous spectrum in  $\ell^2(\rho+\Lambda,\Delta)$, with
the wave functions $\psi_\xi$, $\xi\in \mathbb{A}$ in Eq. \eqref{fk} constituting an orthogonal basis of (generalized) eigenfunctions.

\begin{remark}
For $\hat{t}_2\to 0$ the lattice Hamiltonian $H$ \eqref{Hdiagonal} becomes of the form
\begin{subequations}
\begin{align}\
H=&  \sum_{j=1}^n  \Biggl(  \ \hat{t}_0^{-1} (1-\hat{t}_0\hat{t}_1q^{x_j})\Bigl(\prod_{\substack{1\leq k\leq n \\ k\neq j}} \frac{t^{-1}-q^{x_j-x_k}}{1-q^{x_j-x_k}} \Bigr)   T_j \\
 +& \ \hat{t}_0 (1-q^{x_j})\Bigl(\prod_{\substack{1\leq k\leq n \\ k\neq j}} \frac{t-q^{x_j-x_k}}{1-q^{x_j-x_k}} \Bigr)   T_j^{-1}  +(\hat{t}_0+\hat{t}_1)q^{x_j} \Biggr) -\varepsilon_0 ,\nonumber
\end{align}
with $x=\rho+\lambda$ and
\begin{equation}\label{E0}
\varepsilon_0:=\sum_{j=1}^n  (\hat{t}_0t^{n-j}+\hat{t}_0^{-1}t^{j-n}) .
\end{equation}
\end{subequations}
Indeed, this  readily follows from Eqs. \eqref{H}, \eqref{pms} 
with the aid of the elementary polynomial identity (cf. Example 2. (a) of \cite[Ch. VI.3]{mac:symmetric})
$$
\sum_{j=1}^n (1+z_j)\prod_{\substack{1\leq k\leq n \\ k\neq j}} \frac{t-z_j/z_k}{1-z_j/z_k} =\sum_{j=1}^n (z_j+t^{n-j}) .
$$
\end{remark}

\section{Scattering}\label{sec4}
In this section we rely on results from \cite{die:scattering}, permitting to describe briefly how the $n$-particle scattering operator for the hyperbolic quantum Ruijsenaars-Schneider system on the lattice computed by Ruijsenaars
\cite{rui:factorized} gets modified due to the presence of the external Morse interactions. 
Specifically, the scattering process of the present model with Morse terms turns out to be governed by an
$n$-particle scattering matrix $\hat{ {\mathcal S}} (\xi)$ that factorizes in two-particle and one-particle matrices:
\begin{subequations}
\begin{equation}
\hat{ {\mathcal S}} (\xi)
 := \prod_{1\leq j<k\leq n} s(\xi_j-\xi_k)s(\xi_j+\xi_k)\prod_{1\leq j\leq n} s_0(\xi_j) ,
\end{equation}
with
\begin{equation}
s(x):=\frac{(qe^{ix},te^{-ix})_\infty }{(qe^{-ix},te^{ix})_\infty }\quad\text{and}\quad  
s_0(x):=
\frac{(qe^{2ix})_\infty }{(qe^{-2ix})_\infty }\prod_{0\leq r\leq 2}\frac{(\hat{t}_re^{-ix})_\infty }{(\hat{t}_re^{ix})_\infty },
\end{equation}
\end{subequations}
which compares to Ruijsenaars' scattering matrix $ \prod_{1\leq j<k\leq n} s(\xi_j-\xi_k)$  for the corresponding model without Morse interactions \cite{rui:factorized}.

To substantiate further some additional notation is needed. Let us denote by $\mathcal{H}_0$ the self-adjoint discrete Laplacian in $\ell^2(\Lambda)$ of the form
\begin{equation*}
(\mathcal{H}_0 f)(\lambda)
:=
\sum_{\substack{1\leq j \leq n\\ \lambda+e_j\in\Lambda}} f(\lambda+e_j)
+\sum_{\substack{1\leq j \leq n\\ \lambda-e_j\in\Lambda}}   f(\lambda-e_j)
\qquad (f\in\ell^2(\Lambda)),
\end{equation*}
and let 
\begin{equation}\label{pfH}
\mathcal{H}:= \boldsymbol{\Delta}^{1/2} (H+\varepsilon_0) \boldsymbol{\Delta}^{-1/2},
\end{equation}
with $H$  and $\varepsilon_0$ taken from \eqref{Haction} and \eqref{E0}, respectively. Here the operator
$\boldsymbol{\Delta}^{1/2}:\ell^2(\rho+\Lambda,\Delta)\to \ell^2(\Lambda)$  refers to the Hilbert space isomorphism
\begin{equation}\label{embed1}
(\boldsymbol{\Delta}^{1/2}f)(\lambda):=\Delta^{1/2}_\lambda f(\rho+\lambda)
\qquad (f\in \ell^2(\rho+\Lambda,\Delta) )
\end{equation}
(with $\boldsymbol{\Delta}^{-1/2}:=(\boldsymbol{\Delta}^{1/2})^{-1}$). Then (by Theorem \ref{diagonal:thm})
\begin{equation}\label{F}
\mathcal{H}=\boldsymbol{\mathcal{F}}^{-1}( \hat{E} +\varepsilon_0)\boldsymbol{\mathcal{F}}
\quad
\text{with} 
\quad
\boldsymbol{\mathcal{F}}:= \boldsymbol{\hat{\Delta}}^{1/2} \boldsymbol{F} \boldsymbol{\Delta}^{-1/2} ,
\end{equation}
where $\boldsymbol{\hat{\Delta}}^{1/2}:L^2(\mathbb{A},\hat{\Delta}\text{d}\xi)\to L^2(\mathbb{A})$ denotes the Hilbert space isomorphism\begin{equation}
(\boldsymbol{\hat{\Delta}}^{1/2}\hat{f})(\xi):= \hat{\Delta}^{1/2}(\xi)\hat{f}(\xi)
\qquad (\hat{f}\in L^2(\mathbb{A},\hat{\Delta}\text{d}\xi))
\end{equation}
(and $\hat{E}$ \eqref{E} is now regarded as a self-adjoint bounded multiplication operator in $L^2(\mathbb{A})$).
Furthermore, one has that
$$
\mathcal{H}_0=\boldsymbol{\mathcal{F}}_0^{-1}( \hat{E}+\varepsilon_0)\boldsymbol{\mathcal{F}}_0,
$$
where $\boldsymbol{\mathcal{F}}_0:\ell^2(\Lambda)\to L^2(\mathbb{A})$ denotes the Fourier isomorphism
recovered from $\mathcal{F}$ in the limit $q,t\to 0$,
$\hat{t}_r\to 0$ ($r=0,1,2$). Specifically, this amounts to the Fourier transform
\begin{subequations}
\begin{equation}
(\boldsymbol{\mathcal{F}}_0{f})(\xi)=\sum_{\lambda\in\Lambda}f(\lambda)
\overline{\chi_\xi (\lambda)}
\end{equation}
(${f}\in \ell^2(\Lambda)$) 
with the inversion formula
\begin{equation}
(\boldsymbol{\mathcal{F}}_0^{-1}\hat{f})(\lambda) =
\int_{\mathbb{A}} \hat{f} (\xi) \chi_\xi (\lambda)\text{d}\xi
\end{equation}
\end{subequations}
($\hat{f}\in L^2(\mathbb{A})$) associated with the anti-invariant Fourier kernel
$$
 \chi_\xi (\lambda):=\frac{1}{(2\pi )^{n/2}\, i^{n^2}} \sum_{w\in W} \text{sign}(w) e^{i\langle w(\rho_0 +\lambda ) ,\xi\rangle} ,
$$
where $\text{sign}(w)=\epsilon_1\cdots\epsilon_n\text{sign}(\sigma)$ for $w=(\sigma,\epsilon)\in W=S_n\ltimes \{ 1,-1\}^n$ and 
$\rho_0= (n,n-1,\ldots,2,1)$.

Let $C_0(\mathbb{A}_{\text{reg}})$ be the dense subspace of $L^2(\mathbb{A})$ consisting of smooth test functions with compact support in
 the open dense subset $\mathbb{A}_{\text{reg}}\subset\mathbb{A}$ on which the components of the gradient
$$\nabla \hat{E}(\xi)=(-2\sin(\xi_1),\ldots,-2\sin(\xi_n)),\quad \xi\in\mathbb{A}$$ do not vanish and are all distinct in absolute value.
 We define the following unitary multiplication operator $ \hat{\mathcal S} : L^2(\mathbb{A},\text{d}\xi)\to  L^2(\mathbb{A},\text{d}\xi)$ via its restriction to $C_0(\mathbb{A}_{\text{reg}})$:
 \begin{equation}
 ( \hat{\mathcal S}\hat{f})(\xi):=  \hat{\mathcal S}(w_\xi \xi )\hat{f}(\xi)\qquad
 (\hat{f}\in C_0(\mathbb{A}_{\text{reg}}),
  \end{equation}
 where $w_\xi\in W$ for $\xi\in \mathbb{A}_{\text{reg}}$ is the signed permutation such that the components of $w_\xi \nabla \hat{E}(\xi)$
 are all positive and reordered from large to small.

Theorem~4.2 and Corollary~4.3 of
Ref. \cite{die:scattering} now provide explicit formulas for the wave operators and scattering operator comparing the large-times asymptotics of the interacting particle dynamics
 $e^{i\mathcal{H}t}$ relative to the Laplacian's reference dynamics $e^{i\mathcal{H}_0t}$ as a 
 continuous dual $q$-Hahn reduction of  \cite[Thm. 6.7]{die:scattering}.
 
 \begin{theorem}[Wave and Scattering Operators]\label{scattering:thm}
  The operator limits
\begin{equation*}
\Omega^{\pm} :=s-\lim_{t\to \pm \infty}  e^{i t  \mathcal{H}}e^{-it \mathcal{H}_{0}}
\end{equation*}
converge in the strong $\ell^2(\Lambda)$-norm topology and the corresponding wave operators $\Omega^\pm$ intertwining the interacting dynamics $e^{i\mathcal{H}t}$  with the discrete Laplacian's dynamics $e^{i\mathcal{H}_0t}$
are given by unitary operators in $\ell^2(\Lambda )$ of the form
\begin{equation*}
\Omega^\pm = \boldsymbol{\mathcal{F}}^{-1} \circ \hat{\mathcal S}^{\mp 1/2}  \circ \boldsymbol{\mathcal{F}}_0,
\end{equation*}
where the branches of the square roots are to be chosen such that
\begin{equation*}
s(x)^{1/2}=\frac{(qe^{ix})_\infty }{|(qe^{ix})_\infty | } \frac{|(te^{ix})_\infty | }{(te^{ix})_\infty }
\quad\text{and}\quad
 s_0(x)^{1/2}=\frac{(qe^{2ix})_\infty }{|(qe^{2ix})_\infty | }\prod_{0\leq r\leq 2}\frac{|(\hat{t}_re^{ix})_\infty |}{(\hat{t}_re^{ix})_\infty }  .
 \end{equation*}
The scattering operator relating the large-times asymptotics of  $e^{i\mathcal{H}t}$  for $t\to - \infty$ and $t\to +\infty$ is thus given by the unitary operator
\begin{equation*}
\mathcal{S}:=(\Omega^+)^{-1} \Omega^- =  \boldsymbol{\mathcal{F}}_0^{-1}  \circ \hat{\mathcal S} \circ  \boldsymbol{\mathcal{F}}_0 .
\end{equation*}
\end{theorem}

\section{Bispectral dual system}\label{sec5}
The bispectral dual in the sense of Duistermaat and Gr\"unbaum \cite{dui-gru:differential,gru:bispectral} of the 
hyperbolic quantum Ruijsenaars-Schneider system on the lattice is given  by the trigonometric Ruijsenaars-Macdonald $q$-difference operators 
\cite{rui:complete,mac:symmetric}. This bispectral duality is a quantum manifestation of the duality between the classical Ruijsenaars-Schneider systems with hyperbolic/trigonometric dependence on the position/momentum variables and vice versa
\cite{rui:action}, which (at the classical level) states that the respective action-angle transforms linearizing the two systems under consideration are inverses of each other. As a degeneration of the Macdonald-Koornwinder $q$-difference operator \cite[Eq. (5.4)]{koo:askey-wilson}, we immediately arrive at a bispectral dual Hamiltonian for our hyperbolic quantum Ruijsenaars-Schneider system with Morse term.

Indeed, the continuous dual $q$-Hahn reduction of the $q$-difference equation  satisfied by the Macdonald-Koornwinder polynomials  \cite[Thm. 5.4]{koo:askey-wilson}  reads
\begin{subequations}
\begin{equation}\label{hd1}
\hat{H}P_\lambda = E_\lambda P_\lambda \quad\text{with}\quad E_\lambda= \sum_{j=1}^n t^{j-1}(q^{-\lambda_j}-1)  \qquad (\lambda\in\Lambda),
\end{equation}
where
\begin{equation}\label{hd2}
\hat{H}=\sum_{j=1}^n \left( \hat{v}_j(\xi)(\hat{T}_{j,q}-1)+ \hat{v}_{j}(-\xi)\hat{T}_{j,q}^{-1}-1)\right) ,
\end{equation}
and
\begin{equation}\label{hd3}
\hat{v}_{ j}(\xi)=
\frac{\prod_{0\leq r\leq 2} (1-\hat{t}_re^{i\xi_j}) }{(1-e^{2i\xi_j})  (1-q e^{2i\xi_j}) }
\prod_{\substack{1\leq k\leq n\\k\neq j}}  \frac{1-te^{i(\xi_j+\xi_k)}}{1-e^{i(\xi_j+\xi_k)}}  \frac{
1-te^{i(\xi_j-\xi_k)}}{1-e^{i(\xi_j-\xi_k)}} .
\end{equation}
\end{subequations}
Here $\hat{T}_{j,q}$ acts on trigonometric (Laurent) polynomials $\hat{p}(e^{i\xi_1},\ldots ,e^{i\xi_n})$  by a $q$-shift of the $j$th variable:
\begin{equation*}
(\hat{T}_{j,q}\hat{p})(e^{i\xi_1},\ldots ,e^{i\xi_n}):=\hat{p}(e^{i\xi_1},\ldots,e^{i\xi_{j-1}},qe^{i\xi_j},e^{i\xi_{j+1}},\ldots,e^{i\xi_n}) .
\end{equation*}

In other words, the bispectral dual Hamiltonian $\hat{H}$  \eqref{hd2},\eqref{hd3} constitutes a  nonnegative unbounded self-adjoint operator with purely discrete spectrum in
$L^2(\mathbb{A},\hat{\Delta}\text{d}\xi)$ that is
diagonalized by the (inverse) Fourier transform $\boldsymbol{F}$ \eqref{ft1}, \eqref{ft2}:
 \begin{equation}\label{Hddiagonal}
\hat{H}=\boldsymbol{F}  \circ {{E}} \circ\boldsymbol{F}^{-1},
 \end{equation}
where $E$ denotes the self-adjoint multiplication operator in $ \ell^2(\rho+\Lambda,\Delta)$ of the form
$
(E f)(\rho+\lambda):=E_\lambda f(\rho+\lambda) 
$
 (for $\lambda\in\Lambda$ and $f \in \ell^2(\rho+\Lambda,\Delta)$ with $\langle Ef,Ef\rangle_{\Delta}<\infty$).

\section{Quantum integrability}\label{sec6}
In this final section we provide explicit formulas for a complete system of commuting quantum integrals for the 
hyperbolic quantum Ruijsenaars-Schneider Hamiltonian with Morse term on the lattice $H$ \eqref{Haction} and for its bispectral dual Hamiltonian $\hat{H}$ \eqref{hd2}, \eqref{hd3}.
This confirms the quantum integrability of both Hamiltonians in the present Hilbert space set-up.

\subsection{Hamiltonian}
The quantum integrals for the
hyperbolic Ruijsenaars-Schneider Hamiltonian with Morse term are given by commuting difference operators $H_1,\ldots ,H_n$ that are defined via their action on $f\in \ell^2(\rho+\Lambda ,\Delta)$
(cf. \cite[Eqs. (2.20)--(2.23)]{die:difference}):
\begin{eqnarray}\label{Hint}
\lefteqn{ (H_l f)(\rho +\lambda ) := } && \\
 && \sum_{\substack{J_+, J_-\subset \{1,\ldots,n\}   \\ J_+\cap J_-=\emptyset, \, |J_+|+| J_- | \leq l \\ \lambda + e_{J_+}-e_{J_-}\in\Lambda }     } 
U_{J_+^c\cap J_-^c, l-|J_+|-|J_-|}(\lambda) V_{J_+,J_-}(\lambda )  f(\rho+\lambda+e_{J_+}-e_{J_-})  \nonumber
 \end{eqnarray}
($\lambda\in\Lambda$, $l=1,\ldots ,n$), where $e_J:=\sum_{j\in J} e_j$ for $J\subset \{1,\ldots ,n\}$, $J^c:=\{ 1,\ldots ,n\}\setminus J$ and
 \begin{align*}
 V_{J_+,J_-}(\lambda ) &= t^{-\frac{1}{2}|J_+|(|J_+|-1)+\frac{1}{2}|J_-|(|J_-|-1)}\\
 & \times\prod_{j\in J_+}  \sqrt{\frac{qt_0}{t_1t_2} } (1-{t}_1 t^{n-j}q^{\lambda_j})(1-t_2t^{n-j}q^{\lambda_j}) \\
&\times \prod_{j\in J_-} \sqrt{\frac{t_1t_2}{qt_0}} (1- t_0t^{n-j}q^{\lambda_j})(1-t^{n-j}q^{\lambda_j}) 
 \\
&\times \prod_{\substack{j\in J_+\\k\in J_-}}
\Bigl(\frac{1-t^{1+k-j}q^{\lambda_j-\lambda_k}}{1-t^{k-j}q^{\lambda_j-\lambda_k}} \Bigr)
\Bigl( \frac{t^{-1}-t^{k-j}q^{\lambda_j-\lambda_k+1}}{1-t^{k-j}q^{\lambda_j-\lambda_k+1}}\Bigr)
 \\
&\times \prod_{\substack{j\in J_+\\ k\not\in J_+\cup J_-}} 
\frac{t^{-1}-t^{k-j}q^{\lambda_j-\lambda_k}}{1-t^{k-j}q^{\lambda_j-\lambda_k}} 
\prod_{\substack{j\in J_-\\ k\not\in J_+\cup J_-}} 
 \frac{t-t^{k-j}q^{\lambda_j-\lambda_k}}{1-t^{k-j}q^{\lambda_j-\lambda_k}} ,
  \end{align*} 
 \begin{align*}
U_{K, p}(\lambda)  = (-1)^p&\times \\
\sum_{\substack{I_+,I_-\subset K\\ I_+\cap I_-=\emptyset,\, |I_+|+|I_-|= p }} &
\Biggl(
\prod_{j\in I_+}  \sqrt{\frac{qt_0}{t_1t_2} } (1-{t}_1 t^{n-j}q^{\lambda_j})(1-t_2t^{n-j}q^{\lambda_j}) \\
&\times \prod_{j\in I_-} \sqrt{\frac{t_1t_2}{qt_0}} (1- t_0t^{n-j}q^{\lambda_j})(1-t^{n-j}q^{\lambda_j}) 
 \\
&\times  \prod_{\substack{j\in I_+\\k\in I_-}}
\Bigl(\frac{1-t^{1+k-j}q^{\lambda_j-\lambda_k}}{1-t^{k-j}q^{\lambda_j-\lambda_k}} \Bigr)
\Bigl( \frac{1-t^{-1+k-j}q^{\lambda_j-\lambda_k+1}}{1-t^{k-j}q^{\lambda_j-\lambda_k+1}}\Bigr)\\
&\times \prod_{\substack{j\in I_+\\ k\in K\setminus (I _+\cup I_-)}} 
\frac{t^{-1}-t^{k-j}q^{\lambda_j-\lambda_k}}{1-t^{k-j}q^{\lambda_j-\lambda_k}} 
\prod_{\substack{j\in I_-\\ k\in K\setminus (I _+\cup I_-)}} 
 \frac{t-t^{k-j}q^{\lambda_j-\lambda_k}}{1-t^{k-j}q^{\lambda_j-\lambda_k}}\Biggr) .
  \end{align*}
For $l=1$ the action of $H_l$ \eqref{Hint} is seen to reduce to that of  $H$ \eqref{Haction}.
The diagonalization in Theorem \ref{diagonal:thm} generalizes to these higher commuting quantum integrals as follows.

\begin{theorem}\label{int:thm}
For parameters of the form as in Theorem \ref{diagonal:thm}, the difference operators $H_1,\ldots,H_n$ \eqref{Hint}  constitute  bounded commuting self-adjoint operators in the Hilbert space $\ell^2(\rho+\Lambda,\Delta)$  
that are simultaneously diagonalized by the Fourier transform $\boldsymbol{F}$ \eqref{ft1}, \eqref{ft2}:
\begin{subequations}
 \begin{equation}\label{Hint-diagonal}
H_l=\boldsymbol{F}^{-1}  \circ \hat{{E}}_l  \circ\boldsymbol{F} \qquad (l=1,\ldots ,n),
 \end{equation}
where   $\hat{{E}}_l$ denotes the bounded real multiplication operator acting on
$\hat{f}\in L^2(\mathbb{A},\hat{\Delta}\text{d}\xi)$ by $(\hat{{E}_l}\hat{f})(\xi):=\hat{E}_l(\xi) \hat{f}(\xi)$ with
\begin{align}
&\hat{E}_l (\xi):=\\
&\sum_{1\leq j_1<\cdots < j_l\leq n} (   2\cos(\xi_{j_1}) -t^{j_1-1}\hat{t}_0-t^{-(j_1-1)}\hat{t}_0^{-1} )\cdots (   2\cos(\xi_{j_l}) -t^{j_l-l}\hat{t}_0-t^{-(j_l-l)}\hat{t}_0^{-1} )  .\nonumber
\end{align}
\end{subequations}
\end{theorem}

\begin{proof}
The eigenvalue equation $H_l\psi_\xi=\hat{E}_l(\xi)\psi_\xi$ reads explicitly
\begin{align*}
 \sum_{\substack{J_+, J_-\subset \{1,\ldots,n\}   \\ J_+\cap J_-=\emptyset, \, |J_+|+| J_- | \leq l \\ \lambda + e_{J_+}-e_{J_-}\in\Lambda }     } 
U_{J_+^c\cap J_-^c, l-|J_+|-|J_-|}(\lambda) V_{J_+,J_-}(\lambda )  \psi_\xi (\rho+\lambda+e_{J_+}-e_{J_-})  &\\
=\hat{E}_l(\xi)\psi_\xi (\rho +\lambda) .&
\end{align*}
This eigenvalue identity corresponds to the continuous dual $q$-Hahn reduction of
the Pieri recurrence formula for the Macdonald-Koornwinder polynomials in
\cite[Thm. 6.1]{die:properties}, where we have expressed the eigenvalues $\hat{E}_l(\xi)$ in a compact form stemming from  \cite[Eq. (5.1)]{kom-nou-shi:kernel} (cf. also \cite[Sec. 2.2]{die-ems:branching}).
\end{proof}

\subsection{Bispectral dual Hamiltonian}
The continuous dual $q$-Hahn reduction of the system of higher $q$-difference equations for the Macdonald-Koornwinder polynomials in
 \cite[\text{Sec.}~5.1]{die:properties} reads
 \begin{subequations}
\begin{equation}
\hat{H}_l P_\lambda =E_{\lambda ,l} P_\lambda    \qquad (\lambda\in\Lambda,\, l=1,\ldots ,n),
\end{equation}
where
\begin{equation}\label{eigenvalues}
E_{\lambda ,l} :=t^{-l(l-1)/2} \sum_{1\leq j_1<\cdots <j_l\leq n}
 (t^{j_1-1}q^{-\lambda_{j_1}}-t^{n-j_1})\cdots  (t^{j_l-1}q^{-\lambda_{j_l}}-t^{n+l-j_l-1})
\end{equation}
(cf. \cite[Eq. (5.1)]{kom-nou-shi:kernel}), and
\begin{equation}\label{dual-int}
\hat{H}_l:= 
\sum_{\substack{ J\subset \{ 1,\ldots ,n\},\, 0\leq |J|\leq l \\ \epsilon_j\in\{ 1,-1\},j\in J}} \hat{U}_{J^c,l -|J|}\hat{V}_{\epsilon J} \hat{T}_{\epsilon J,q},
\end{equation}
with $\hat{T}_{\epsilon J,q}:=\prod_{j\in J} \hat{T}_{j,q}^{\epsilon_j}$ and
\begin{align*}
\hat{V}_{\epsilon J} = &
\prod_{j\in J}  \frac{\prod_{0\leq r\leq 2} (1-\hat{t}_re^{i\epsilon_j\xi_j}) }{(1-e^{2i\epsilon_j\xi_j})  (1-q e^{2i\epsilon_j\xi_j}) } \prod_{\substack{j\in J \\k\not\in J}} \frac{1-te^{i(\epsilon_j\xi_j+\xi_k)}}{1-e^{i(\epsilon_j\xi_j+\xi_k)}}\frac{1-te^{i(\epsilon_j\xi_j-\xi_k)}}{1-e^{i(\epsilon_j\xi_j-\xi_k)}}
 \\
& \times \prod_{\substack{j,k\in J \\ j<k}} \frac{1-te^{i(\epsilon_j\xi_j+\epsilon_k\xi_k)}}{1-e^{i(\epsilon_j\xi_j+\epsilon_k\xi_k)}}
\frac{1-tqe^{i(\epsilon_j\xi_j+\epsilon_j\xi_k)}}{1-qe^{i(\epsilon_j\xi_j+\epsilon_k\xi_k)}}  ,
\end{align*}
\begin{align*}
\hat{U}_{K,p} =  (-1)^p 
\sum_{\substack{ I \subset K,\,  | I | = p\\ \epsilon_j\in\{ 1,-1\},j\in I}} 
& \bigg( 
 \prod_{j\in I}  \frac{\prod_{0\leq r\leq 2} (1-\hat{t}_re^{i\epsilon_j\xi_j}) }{(1-e^{2i\epsilon_j\xi_j})  (1-q e^{2i\epsilon_j\xi_j}) } \\
&\times   \prod_{\substack{j\in I \\k\in K\setminus I}} \frac{1-te^{i(\epsilon_j\xi_j+\xi_k)}}{1-e^{i(\epsilon_j\xi_j+\xi_k)}}\frac{1-te^{i(\epsilon_j\xi_j-\xi_k)}}{1-e^{i(\epsilon_j\xi_j-\xi_k)}} \\
 & \times \prod_{\substack{j,k\in I \\ j<k}} \frac{1-te^{i(\epsilon_j\xi_j+\epsilon_k\xi_k)}}{1-e^{i(\epsilon_j\xi_j+\epsilon_k\xi_k)}}
\frac{t-q e^{i(\epsilon_j\xi_j+\epsilon_j\xi_k)}}{1-qe^{i(\epsilon_j\xi_j+\epsilon_k\xi_k)}}  \bigg) .
\end{align*}
 \end{subequations}
For $l=1$, this reproduces the  continuous dual $q$-Hahn reduction of the Macdonald-Koornwinder $q$-difference equation in Eqs. \eqref{hd1}--\eqref{hd3}.

The $q$-difference operators $\hat{H}_1,\ldots,\hat{H}_n$ extend the bispectral dual Hamiltonian $\hat{H}$ \eqref{hd2}--\eqref{hd3}
into a complete system of commuting quantum integrals that   are simultaneously diagonalized by the
multivariate continuous dual $q$-Hahn polynomials.
\begin{theorem}\label{dual:thm} 
For parameter values in the domain \eqref{pdomain}, the $q$-difference operators $\hat{H}_1,\ldots,\hat{H}_n$ constitute {\em nonnegative} unbounded self-adjoint operators with purely discrete spectra in
$L^2(\mathbb{A},\hat{\Delta}\text{d}\xi)$ that are simultaneously diagonalized by the (inverse) Fourier transform $\boldsymbol{F}$ \eqref{ft1}, \eqref{ft2}:
\begin{equation}
\hat{H}_l=\boldsymbol{F}  \circ {{E}}_l \circ\boldsymbol{F}^{-1},\qquad l=1,\ldots ,n,
 \end{equation}
where $E_l$ denotes the self-adjoint multiplication operator in $ \ell^2(\rho+\Lambda,\Delta)$ given by
$(E_l f)(\rho+\lambda):=E_{\lambda ,l} f(\rho+\lambda)$
(on the domain of $f \in \ell^2(\rho+\Lambda,\Delta)$ such that $\langle E_l f,E_l f\rangle_{\Delta}<\infty$).
\end{theorem}
Notice in this connection that although the domain of the unbounded operator $\hat{H}_l$ in $L^2(\mathbb{A},\hat{\Delta}\text{d}\xi)$  depends on $l$, the resolvent operators $(\hat{H}_1-z_1)^{-1},\ldots ,(\hat{H}_n-z_n)^{-1}$ (with $z_1,\ldots ,z_n\in\mathbb{C}\setminus [0,+\infty)$) commute as bounded operators on $L^2(\mathbb{A},\hat{\Delta}\text{d}\xi)$, and
the $q$-difference operators $\hat{H}_1,\ldots ,\hat{H}_n$ moreover commute themselves on the joint polynomial eigenbasis $P_\lambda$, $\lambda\in \Lambda$.

\begin{remark} To infer that the eigenvalues $E_{\lambda ,l}$ \eqref{eigenvalues} are nonnegative---thus indeed giving rise to a {\em nonnegative} operator $\hat{H}_l$ in Theorem \ref{dual:thm}---it is helpful to note that  these can be rewritten as (cf.  \cite[\text{Sec.}~5.1]{die:properties})
$$ E_{\lambda ,l}=t^{-l(l-1)/2}
E_{l,n}(q^{-\lambda_1},tq^{-\lambda_2},\ldots, t^{n-1}q^{-\lambda_n}; t^{l-1},t^l,\ldots ,t^{n-1})$$
with
$$E_{l,n} (z_1,\ldots ,z_n;y_l,\ldots ,y_n):= \sum_{0\leq k\leq l} (-1)^{l+k} \boldsymbol{e}_k(z_1,\ldots ,z_n) \boldsymbol{h}_{l-k} (y_l,\ldots, y_n).$$
Here $\boldsymbol{e}_k (z_1,\ldots,z_n)$ and $\boldsymbol{h}_k(y_l,\ldots ,y_n)$ refer to the elementary and the complete symmetric functions of degree $k$  (cf. \cite[Ch. I.2]{mac:symmetric}), with the convention that $\boldsymbol{e}_0=\boldsymbol{h}_0\equiv 1$. The nonnegativity of the eigenvalues now readily follows inductively in the particle number $n$ by means of
 the recurrence (cf. \cite[Lem. B.2]{die:commuting})
\begin{align*}
E_{l,n}(q^{-\lambda_1},tq^{-\lambda_2},&\ldots, t^{n-1}q^{-\lambda_n}; t^{l-1},t^l,\ldots ,t^{n-1})=\\
(q^{-\lambda_1}-t^{l-1})
&E_{l-1,n-1}(tq^{-\lambda_2},\ldots, t^{n-1}q^{-\lambda_n}; t^{l-1},\ldots ,t^{n-1})\\
+\ &
E_{l,n-1}(tq^{-\lambda_2},\ldots, t^{n-1}q^{-\lambda_n}; t^l,\ldots ,t^{n-1}) 
\end{align*}
and the homogeneity 
$$
E_{l,n}(tz_1,\ldots ,tz_n; ty_l,\ldots ,ty_n) \\
=t^l
E_{l,n}(z_1,\ldots,z_n;y_l,\ldots ,y_n) .
$$
\end{remark}

\begin{remark}
The hyperbolic Ruijsenaars-Schneider Hamiltonian with Morse term \eqref{H} can be retrieved as a limit of the Macdonald-Koornwinder $q$-difference operator
\cite{die:difference}. In this limit the center-of-mass is sent to infinity, which causes the hyperoctahedral symmetry of the Macdonald-Koornwinder operator  to be broken: 
while the permutation-symmetry still persists the parity-symmetry is no longer present. 
Indeed, the limit in question restores the translational-invariance of the
interparticle pair-interactions enjoyed by the original Ruijsenaars-Schneider model and gives moreover rise
to additional Morse terms that are not parity-invariant. 
It turns out that most of our results above can in fact be lifted to the Macdonald-Koornwinder level, even though 
such a generalization is presumably somewhat less relevant from a physical point of view.
Specifically, the scattering of the corresponding quantum lattice model associated with the full six-parameter family of Macdonald-Koornwinder polynomials
was briefly discussed in \cite[Sec. 6.4]{die:scattering}, its commuting quantum integrals can be read-off from the Pieri formulas for the Macdonald-Koornwinder polynomials in \cite[Thm. 6.1]{die:properties}, and the pertinent bispectral dual Hamiltonian and its commuting quantum integrals are given by the Macdonald-Koornwinder $q$-difference operator \cite{koo:askey-wilson} and its higher-order commuting $q$-difference operators
\cite[Thm. 5.1]{die:properties}.
\end{remark}

\bibliographystyle{amsplain}

\end{document}